\documentclass[11pt]{elsarticle}
\usepackage{stmaryrd}
\usepackage{amsmath,amssymb,amsthm}
\usepackage[colorlinks,linkcolor=red,anchorcolor=blue,citecolor=green]{hyperref}
\usepackage{tikz}
\usepackage{subfigure}
\usepackage{graphicx}

\theoremstyle{plain}
\newtheorem{theorem}{Theorem}[section]
\newtheorem{lemma}[theorem]{Lemma}

\newtheorem{proposition}[theorem]{Proposition}

\theoremstyle{definition}
\newtheorem{definition}[theorem]{Definition}
\newtheorem{remark}[theorem]{Remark}
\newtheorem{example}[theorem]{Example}

\numberwithin{equation}{section}
\newcommand{\dda}{\mathord{\mbox{\makebox[0pt][l]{\raisebox{-.4ex}{$\downarrow$}}$\downarrow$}}}

\newcommand{\da}{\mathord{\downarrow}}
\newcommand{\rom}[1]{\rm{\uppercase\expandafter{\romannumeral #1}}}

\newcommand{\set}[2]{\{#1\mid#2\}}

\makeatletter
\def\ps@pprintTitle{%
  \let\@oddhead\@empty
  \let\@evenhead\@empty
  \def\@oddfoot{\reset@font\hfil\thepage\hfil}
  \let\@evenfoot\@oddfoot
}
\makeatother

\begin{document}

\begin{frontmatter}

\title{ Representations of  Continuous Domains and Continuous $L$-Domains based on
 Closure Spaces\tnoteref{t1}}
\tnotetext[t1]{Supported by the National Natural Science Foundation of China(11771134).}

\author[]{Longchun Wang$^{a,b}$}
\address[1]{School of Mathematical Sciences, Qufu Normal University, Qufu, 273165, China}
\address[2]{School of Mathematics, Hunan University, Changsha, Hunan, 410082, China}
\author[2]{Qingguo Li\corref{a1}}
\cortext[a1]{Corresponding author.}
\ead{liqingguoli@aliyun.com}

\author[]{Lankun Guo$^{c}$}
\address[3]{College of Mathematics and Computer Sciences, Hunan Normal University, Changsha, Hunan, 410081, China}

\begin{abstract}

Closure space has proven to be a useful tool to restructure lattices and various order structures. This paper aims to provide a novel approach to characterizing some important kinds of continuous domains by means of closure spaces.
By introducing an additional map into a given closure space, the notion of $\mathbb{F}$-augmented generalized closure space is presented. It is shown that $\mathbb{F}$-augmented generalized closure spaces
generate exactly continuous domains. Moreover, the notion of approximable mapping is identified to represent Scott-continuous functions between continuous domains. These results produce a category  equivalent to that of continuous domains with Scott-continuous functions.
At the same time, two Cartesian closed categories
 are considered which are equivalent to categories of continuous $L$-domains and continuous bounded-complete domains, respectively.
\end{abstract}

\begin{keyword}
continuous domains \sep continuous $L$-domains \sep $\mathbb{F}$-augmented generalized closure space\sep closure operator \sep categorical equivalence
%\MSC 06B35\sep 18B35\sep 22A26\sep 68Q55
\end{keyword}
\end{frontmatter}
\section{Introduction}

Domain theory is a multi-disciplinary field with wide applications in theoretical computer science, and with deep roots in the mathematical theory
of orders and topologies~\cite{erne18, gierz03, zhao}.
It provides a framework modeling  definitions of  approximation, iteration and computation. In order to give a concrete and accessible alternative for abstract structures of domains, many
approaches have been introduced for  the representations of domains, such as  information systems~\cite{hoofman93, huang16, speen08, Zhang94}, abstract bases~\cite{Lawson97, xu08}, domain logic~\cite{abramsky87a, jung99}, etc.

Closure space is also a useful interdisciplinary tool, consisting of a set and a closure operator on it. %A closure operator  has an one-to-one correspondence to a closure system, the set of fixed-points of the closure operator.
A dual notion of closure space is that of interior space. In a closure space or an interior space, the main mathematical objects investigated are those that can be generated by application of iterative algorithms. It is well-known that the lattice of closed sets of a closure space is complete and every complete lattice arises this way up to isomorphism. The idea of representing order structures by a closure space plus some additional conditions can be traced back to Birkhoff's representation theorem for finite distributive lattices~\cite{birkhoff37} and Stone's duality for Boolean algebras~\cite{stone36}. These famous works allow  us to better understand the nature of  lattices and the interrelationship  between lattices  and closure spaces. Now, the interrelationship and interaction between closure spaces and other mathematical structures  has developed into an important area of the mathematical study  with a thriving theoretical community~\cite{erne04, erne09,  ranzato99}.

In this paper, we  seek to understand the interrelationship between continuous domains and closure spaces. It is worth mentioning that there are already some  systematic investigations of the relationship between algebraic domains and closure spaces. The material on algebraic lattices is classical: it is studied, particularly, in~\cite{davey02, johnstone82}. And in~\cite{davey02}, a  representation of algebraic bounded-complete domains based on closure spaces was also discussed.
 Zhang et al. ~\cite{hitzler06, zhang06} provided a characterization of algebraic lattices in terms of formal concept analysis, an alternative form of closure space. Guo and Li~\cite{guo15} presented a notion of $\mathbb{F}$-augmented closure spaces by adding a structure into a closure space, which realizes its links to algebraic domains. However, all these formalisms only discuss the subclasses of algebraic domains, and the intimate relationship between the category of continuous domains with Scott-continuous functions and closure spaces was not explicated until the current paper.
 From a categorical viewpoint, the category of algebraic domains with Scott continuous functions is a full subcategory of that of continuous domains. There is only  one step  from algebraic domains to continuous domains. Then to obtain a representation of  continuous domains, it should take more objects into account.

In a topological space $(X,\mathcal{O}(X))$, there naturally exist a closure operator and an interior operator.  A subset of $X$ is called regular open if it equals to the interior of its
closure.  This class of sets is found to
have applications not only in topological structures but also in other mathematical structures~\cite{Andretta13, Givant09}. In addition,  clopen sets as the images of the composition of closure and interior operators  play an important role in various Stone's dualities~\cite{erne04, ledda}. %Alternatively, a regular open set is a fixed-point of the composition map~$\tau\circ\gamma$.
Motivated by the above observations, we add a  generalized interior operator~$\tau$ into a given closure space~$(X,\gamma)$ and introduce the notion of  generalized closure space, which generalizes the notions of closure space and interior space.  In some sense, the newly added map $\tau$ together with the closure operator~$\gamma$ in a  generalized closure space has the same effects as the topological interior with topological closure in a topological space. %We call each fixed-point of the composition~$\tau\circ\gamma$ a regular open set of generalized closure space.
  Then following the method used in~\cite{guo15}, we define our main notion of $\mathbb{F}$-augmented generalized closure space, and show it is a concrete representation of continuous domains. This supplies the object level of a functor. Moreover, we identify an appropriate notion of morphism for $\mathbb{F}$-augmented generalized closure spaces, that represents the Scott-continuous functions between continuous domains and produces a category  equivalent to that of continuous domains.

Continuous $L$-domains and continuous bounded-complete domains are two cartesian closed full subcategories of continuous domains. They both are good candidates for denotational semantics of programming languages. For these two kinds of domains, this paper also studies how they can be represented by $\mathbb{F}$-augmented generalized closure spaces. Our results demonstrate the capacity of closure spaces in representing various continuous domains.

The content is arranged as follows. Section 2 recalls some basic notions of posets and domains  needed in the
sequel. Section 3 focuses on object-level correspondences between continuous domains and $\mathbb{F}$-augmented generalized closure spaces. The cases of continuous $L$-domains and continuous bounded-complete domains are also discussed. The main contribution of Section~4 is  the introduction of a notion of approximable mapping on $\mathbb{F}$-augmented generalized closure spaces which produces a
category equivalent to that of continuous domains.% with Scott-continuous functions.
\section{Preliminaries}

In this section, we fix some terminology and
 recall some  definitions and results that  will be used to develop our theory.

For any set $X$, the symbol $F\sqsubseteq X$  means that $F$ is a finite subset of $X$.  $\mathcal{P}(X)$ and $\mathcal{F}(X)$ stand for the powerset of $X$ and the finite powerset of $X$, respectively.

 If $A$ is a subset of a given poset~$(P,\leq)$, then $\da A$ is defined to be the set $\set{x\in P}{\exists a\in A, x\leq a}$. For short, we  write $\da x$ for $\da \{x\}$. A subset  $D$ of $P$ is called \emph{directed} if it is nonempty and every finite subset of it has an upper bound in $D$. A \emph{dcpo} is a poset $P$  such that each directed subset~$D$ of  $P$ has a  sup~$\bigvee D$ in  $P$. %A poset $P$ is called a \emph{$\vee$-semilattice} if and only if every pair of elements of $P$ has a sup.
 A  \emph{complete lattice} is a poset $P$ in which every subset has a sup.
%We use $\bigsqcap X$ to denote the greatest lower bound of $X$, and $\bigsqcup X$ to denote the least upper bound of $X$, respectively. In the case of pairs of elements it is customary to write $x\sqcap y =\bigsqcap \{x,y\}$. The operation $\sqcap$ is often called meet. A semi-lattice is a poset S in which any two elements $a, b$ have a meet~$a\sqcap b$.

For any elements $x,y$  of a dcpo $P$, $x$ is \emph{way below} $y$, in symbols $x\ll y$, if and only if for any directed subset $D$ of $P$ with  $y\leq \bigvee D$ always yields that $x\leq d$ for some $d\in D$. If $x$ and $y$ are elements of dcpo $P$ with $x\ll y$, then $x\leq y$.
 For any subset $A\subseteq P$, we write $\dda A=\set{x\in P}{\exists a\in A, x\ll a}$ and $\dda x=\dda\{x\}$. A subset $B_P\subseteq P$ is said to be a \emph{basis} of $P$ if, for any element $x\in P$, there is a directed subset $D\subseteq\dda x \cap B_P $ such that $x=\bigvee D$.
   An element $x\in P$ is called \emph{compact} if $x\ll x$. The set of compact elements of $P$ is denoted by $K(P)$.
   \begin{definition}\label{d2.1}
   \begin{enumerate}[(1)]
  \item A dcpo $P$ is a \emph{continuous domain} if  it has a basis.
  \item A continuous domain is said to be a \emph{continuous $L$-domain} if, for any $x\in P$, the set $\da x$ is a complete lattice in its induced order.
  \item A \emph{continuous bounded-complete domain} is a continuous domain in which every bounded above subset has a sup.
  \item A  dcpo $P$ is said to be an \emph{algebraic domain} if $K(P)$ forms a basis.
  \item An algebraic domain which is a complete lattice is called an \emph{algebraic lattice}.
   \end{enumerate}
   \end{definition}

Particularly, every algebraic domain is a continuous domain. Note that if
  $P$ is a continuous domain  with a basis $B_P$, then the way below relation on $P$  satisfies the following \emph{interpolation property}:
$$(\forall M\sqsubseteq P,\forall x\in P)M\ll x\Rightarrow (\exists y\in B_P)M\ll y\ll x,$$
where $M\ll x$ means $m\ll x$ for all elements $m\in M$.

\begin{definition}\label{d2.2}
A function~$f:P\rightarrow Q$ between continuous domains $P$ and $Q$ is called \emph{Scott-continuous} if for all directed subsets $D$ of $P$, $f(\bigvee D)=\bigvee f(D)$.
\end{definition}

A \emph{closure space} is a pair $(X,\gamma)$ consisting of a set $X$ and a \emph{closure operator}~$\gamma$ on $X$, that is a map on $\mathcal{P}(X)$ such that, for any $A,B\subseteq X$,
\begin{enumerate} [(1)]
\item  $extensive:~A\subseteq\gamma(A)$,
\item $idempotent:~\gamma(\gamma(A))=\gamma(A)$, and
\item $monotone:~\gamma(A)\subseteq\gamma(B)$ whenever $A\subseteq B$.
\end{enumerate}
%The set Fix$(\gamma)$ is called a \emph{closure system}, where Fix$(\gamma)$ stand for the set of all fixed-points of $\gamma$.
 In the same manner, an \emph{interior space} is a pair $(X,\tau)$, where $\tau$ is an \emph{interior
operator} which is monotone, idempotent and \emph{contractive}, i.e., $ \tau(A) \subseteq A$.

A closure operator $\gamma$ on $X$ is \emph{algebraic} if $\gamma(A)=\bigcup\set{\gamma(F)}{F\sqsubseteq A}$ for any $A\subseteq X$. In this case, $(X,\gamma)$ is called an \emph{algebraic closure space}.

%\begin{definition}\cite{guo15}Let $(X,\gamma, \mathcal{F})$ be an $F$-augmented closure space. A subset $Q$ is called an F-closed set of $(X,\gamma, \mathcal{F})$ if, for any $M\sqsubseteq Q$, there exists some $F\in \mathcal{F}$ such that $M\subseteq \gamma(F)\subseteq Q$.\end{definition}

We refer the reader to~\cite{ davey02, gierz03, goubault13a} for standard definitions of order theory and domain theory and~\cite{erne09} for an introduction for closure spaces.

\section{$\mathbb{F}$-augmented generalized closure space}

%Our main purpose is to present a concrete representation of continuous domains in term of closure space. %Following the ideas introduced in~\cite{birkhoff37, erne09, guo15},
In this section, we generalize the notion of closure space and use it to represent continuous domains, continuous $L$-domains and continuous bounded-complete domains. %Roughly speaking, given a closure space $(X,\gamma)$, we first introduce a so-called partial interior operator $\tau$ associated with it, and then obtain a  generalized closure space. Secondly, by selecting an appropriate finite subset of $X$,  we  present the first main notion of the paper, that of $F$-augmented generalized closure space. Finally, on the basis of this notion, we find an concrete representation of continuous domains, as well as a representation of continuous $L$-domains.

\subsection{The representation of continuous domains}

%In this subsection, we introduce the notion of $F$-augmented generalized closure space, which will be used to represent continuous domains.

 Considering a topological space $(X,\mathcal{O}(X))$, along with the topological closure also goes the topological interior. These two operators, together, can be used to create many important mathematical structures. For  example, Macneille and Tarski have shown that the set  of all regular open sets of $(X,\mathcal{O}(X))$ ordered by set inclusion forms a complete Boolean algebra, where a subset of $X$ is \emph{regular open} if and only if it coincides with  the interior of its own closure~\cite{Givant09}.
 The above construction suggests a technique of discussing  the composition map of a closure operator and an interior operator on a fixed underlying set. Inspired by this idea, we make the following definition.
\begin{definition}\label{d3.1} Let $X$ be a set.
 A pair $(X, \langle~\rangle)$ is called a \emph{generalized closure space}, if there are a closure operator $\gamma$ on $X$ and a map $\tau$ on $\mathcal{P}(X)$ such that $\langle~\rangle$ is the composition $\tau\circ\gamma$  which satisfies the following conditions, for all $A,B\subseteq X$
\begin{enumerate}[(1)]
\item $\langle A\rangle\subseteq \gamma(A)$;
\item $\tau(\langle A\rangle)=\langle A\rangle$;
\item $\langle A\rangle\subseteq \langle B\rangle$ whenever $A \subseteq B$.
\end{enumerate}
%Each fixed-point of the composition $\tau\circ\gamma$ is called a \emph{regular open} set of  $(X, \tau\circ\gamma)$.
\end{definition}

While developing the theory, use of $\langle~\rangle$ allows us some
notational simplifications. However in the case that we mention a generalized closure space, we choose to work
with $(X,\tau\circ\gamma)$ since we want both $\gamma$ and $\tau$  to be visible.

\begin{remark}\label{m3.2} It must  to be noted that the identity map on the powerset $\mathcal{P}(X)$ is both a closure operator and an interior operator.
  For any closure space~$(X, \gamma)$, if $\tau$ is the identity map on $\mathcal{P}(X)$, then $\tau\circ\gamma=\gamma$. Thus each closure space $(X, \gamma)$ is a  generalized closure space. Similarly, each interior space is also a generalised closure space. So, the notion of generalized closure space is a generalization of closure spaces as well as interior spaces.

Moreover, let $(X, \gamma)$ be a closure space, and  let $\tau$ be an interior operator on $X$. Then the pair $(X, \tau\circ\gamma)$ is a generalized closure space.

\end{remark}
\begin{example}\label{ex.3.3}

 Given a topological space $(X, \mathcal{O}(X))$, let $c$ and $i$ be the corresponding topological closure and topological interior, respectively.
 It is evident that $(X,c)$, $(X,i)$ and $(X,i\circ c)$ are all generalized closure spaces. In general, $(X,c\circ i)$ is not a generalized closure space.
 %If $\mathcal{O'}(X)$ is another topology on $X$, then $(X,i\circ c')$ and $(X,i'\circ c)$ are also generalized closure spaces.
\end{example}

%In the sequel, we often write $\langle A\rangle$ for $\tau(\gamma(A))$ when the operators~$\gamma$ and $\tau$ are clear.

\begin{proposition}\label{p3.4}
Let $(X,\tau\circ\gamma)$ be a generalized closure space.
\begin{enumerate}[(1)]
 \item If $A\subseteq X$, then $\langle\langle A \rangle \rangle\subseteq\langle A \rangle$.
     \item If $A,B\subseteq X$ and $A\subseteq \langle B \rangle$, then $\langle A \rangle\subseteq \langle B \rangle$.
\end{enumerate}
 \end{proposition}

\begin{proof} (1)  By Definition  \ref{d3.1}, we have $$\langle\langle A \rangle \rangle\subseteq\langle\gamma(A) \rangle
=\tau(\gamma(\gamma(A)))=\tau(\gamma(A))=\langle A \rangle.$$

(2) Straightforward from Definition  \ref{d3.1} and part (1).
\end{proof}
%It can easily be seen that $\set{\langle A\rangle}{A\subseteq X}$ is just  the set of all fixed-points of $\tau\circ\gamma$. Moreover, $\set{\langle A\rangle}{A\subseteq X}$ ordered by set inclusion forms a complete lattice, since $\tau\circ\gamma$ is a monotone map on $\mathcal{P}(X)$. In addition, Proposition~\ref{p3.4} also holds for any generalized closure spaces. To establish a bridge from $\mathbb{F}$-augmented generalized closure spaces to continuous domains, we need a further definition.

In \cite{guo15}, Guo and Li introduced a notion of $\mathbb{F}$-augmented closure space which consists of a closure space $(X,\gamma)$ and a nonempty subset $\mathcal{F}\subseteq\mathcal{F}(X)$ satisfying some additional information.

\begin{definition}\cite{guo15}\label{d3.5}
Let $(X,\gamma)$ be a closure space and $\mathcal{F}$ a nonempty family of finite subsets of $X$. The triplet $(X,\gamma, \mathcal{F})$ is called an \emph{finite-subset-selection augmented closure space (for short, $\mathbb{F}$-augmented closure space)} if, for any $F\in \mathcal{F}$ and $M\sqsubseteq \gamma(F)$, there exists $F_1\in \mathcal{F}$ such that $M\subseteq F_1\subseteq \gamma(F)$.
\end{definition}

A main result of \cite{guo15} is that:
for any $\mathbb{F}$-augmented closure space~$(X,\gamma,\mathcal{F})$, relying on $\gamma$ and $\mathcal{F}$ one can single out
some special subsets of $X$. These sets ordered by set inclusion  generate an algebraic domain, and every algebraic domain can be obtained in this way.

Next, we introduce the notion of  $\mathbb{F}$-augmented generalized closure space, by which we provide a concrete representation for continuous domains. This develops the representations of algebraic lattices by  algebraic closure spaces and of algebraic domains by $\mathbb{F}$-augmented closure spaces.

\begin{definition}\label{d3.6}
Let $(X,\tau\circ\gamma)$ be a generalized closure space and $\mathcal{F}$ a nonempty family of finite subsets of X. The triplet $(X,\tau\circ\gamma, \mathcal{F})$ is called an \emph{$\mathbb{F}$-augmented generalized closure space} if, for any $F\in \mathcal{F}$ and $M\sqsubseteq \langle F\rangle$, there exists $F_1\in \mathcal{F}$ such that $M\subseteq \langle F_1\rangle$ and  $M\subseteq F_1\subseteq\langle F\rangle$.
\end{definition}
\begin{example}\label{ex3.7}
Let $\mathbb{R}$ be the real numbers. For any $A\subseteq \mathbb{R}$, define $\gamma (A)=\set{x\in \mathbb{R}}{\exists a\in A, x\geq a}$, and $\tau (A)=\set{x\in \mathbb{R}}{\exists a\in A, x> a}$. Then $(\mathbb{R},\tau\circ\gamma)$ is a generalized closure space. Taking $\mathcal{F}_{\mathbb{R}}=\mathcal{F}(\mathbb{R})$, we obtain an $\mathbb{F}$-augmented generalized closure space~$(\mathbb{R},\tau\circ\gamma,\mathcal{F}_{\mathbb{R}})$.
\end{example}
\begin{example}\label{ex3.8}
If $(X,\gamma)$ is an algebraic closure space and $\tau$ is the identity map on $\mathcal{P}(X)$, then $\tau\circ\gamma=\gamma$. Trivial checks verify that $(X,\gamma, \mathcal{F}(X))$ is an $\mathbb{F}$-augmented closure space.
\end{example}

In Example \ref{ex3.8}, the identity may $\tau$ and the finite powerset $\mathcal{F}(X)$ are trivial.  So each algebraic closure space is naturally  an $\mathbb{F}$-augmented generalized closure space in some sense. In other words, the notion of $\mathbb{F}$-augmented generalized closure spaces is a generalization of algebraic closure spaces. The next result tells us that the notion of $\mathbb{F}$-augmented generalized closure spaces is also a generalization of $\mathbb{F}$-augmented closure spaces.
\begin{proposition}\label{p3.9}
Each $\mathbb{F}$-augmented closure space is an $\mathbb{F}$-augmented generalized closure space.
\end{proposition}
\begin{proof} Let $(X,\gamma, \mathcal{F})$ be an $\mathbb{F}$-augmented closure space. By Remark~\ref{m3.2}, the corresponding closure space $(X,\gamma)$ is a generalised closure space and $\langle A\rangle=\gamma(A)$ for any $A\subseteq X$. Assume that $F\in\mathcal{F}$ and $M\sqsubseteq \langle F\rangle$. Then by Definition~\ref{d3.5}, there exists some $F_1\in \mathcal{F}$ such that $M\subseteq F_1\subseteq\gamma(F)=\langle F\rangle$. Since $\gamma$ is a closure operator and $M\subseteq F_1$, it follows that $M\subseteq F_1\subseteq\gamma(F_1)=\langle F_1\rangle$. Therefore, $(X,\gamma, \mathcal{F})$ is an $\mathbb{F}$-augmented generalized closure space.
\end{proof}

\begin{definition}\label{d3.10}
Let $(X,\tau\circ\gamma, \mathcal{F})$ be an $\mathbb{F}$-augmented generalized closure space. A  subset $U$ of $X$ is called an \emph{$\mathbb{F}$-regular open set} of $(X,\tau\circ\gamma, \mathcal{F})$ if, for any $M\sqsubseteq U$, there exists some $F\in \mathcal{F}$ such that $M\subseteq \langle F\rangle\subseteq U$.
\end{definition}
In the sequel, we use $\mathcal{R}(X)$ to denote the family of all $\mathbb{F}$-regular open sets of $(X,\tau\circ\gamma, \mathcal{F})$.% when $\tau\circ\gamma$ and $\mathcal{F}$ are clear. %The next proposition show that there are enough many $F$-regular open sets.

\begin{proposition}\label{p3.11}
Let $(X, \tau\circ\gamma,\mathcal{F})$ be an $\mathbb{F}$-augmented generalized closure space.
 \begin{enumerate}[(1)]
 \item For any $F\in\mathcal{F}$, $\langle F\rangle$ is an $\mathbb{F}$-regular open set of $(X, \tau\circ\gamma,\mathcal{F})$.
 \item If $U\in \mathcal{R}(X)$ and  $M\sqsubseteq U$, then there exists some $F\in\mathcal{F}$ such that $M\subseteq F\subseteq U$ and $M\subseteq\langle F\rangle \subseteq U$. Hence,  $\langle M\rangle \subseteq U$.
 \end{enumerate}
\end{proposition}

\begin{proof}
Part (1) is clear by Definition~\ref{d3.10}.

For part~(2), let $U\in \mathcal{R}(X)$ with  $M\sqsubseteq U$. Then $M\subseteq \langle F_1\rangle \subseteq U$ for some $F_1\in\mathcal{F}$. By Definition~\ref{d3.6}, there exists some $F\in\mathcal{F}$ such that  $M\subseteq F \subseteq \langle F_1\rangle $ and $M\subseteq \langle F\rangle$. Thus $M\subseteq F\subseteq U$ and $M\subseteq\langle F\rangle \subseteq U$.
\end{proof}
% To lighten the notion of $F$-regular open set of $(X, \tau\circ\gamma,\mathcal{F})$ more, we give a characterization of it in the following.

\begin{proposition}\label{p3.12}
Let $(X, \tau\circ\gamma,\mathcal{F})$ be an $\mathbb{F}$-augmented generalized closure space, and let $U$ be a subset of $X$. Then the following are equivalent.
\begin{enumerate}[(1)]
\item $U$  is an $\mathbb{F}$-regular open set of $(X, \tau\circ\gamma,\mathcal{F})$.
\item The set $\set{\langle F\rangle}{F\in\mathcal{F},F\subseteq U}$ is directed and $U$ is its union.
\end{enumerate}
\end{proposition}
\begin{proof}
(1) implies (2): Because the empty set $\emptyset$ is a subset of $U$, there is some $F\in\mathcal{F}$ such that $M\subseteq F\subseteq U$ by part (2) of Proposition \ref{p3.11}, which implies that the set $\set{\langle F\rangle}{F\in\mathcal{F},F\subseteq U}$ is nonempty. Assume that $F_i\in \mathcal{F}$ and $F_i\subseteq U$, for $i\in\{1,2\}$. Then  $F_1\cup F_2$ is a finite subset of $U$. Since $U$ is an $\mathbb{F}$-regular open set, there exists some $F_3\in \mathcal{F}$ such that $F_1\cup F_2\subseteq \langle F_3\rangle\subseteq U$. To the above $F_1\cup F_2$ and $F_3$, by Definition~\ref{d3.6}, we obtain some $F_4\in \mathcal{F}$ such that $F_1\cup F_2\subseteq \langle F_4\rangle$ and $ F_4\subseteq\langle F_3\rangle\subseteq U$. By part~(2) of Proposition~\ref{p3.4}, it follows that $\langle F_1\rangle\subseteq \langle F_4\rangle$ and $\langle F_2\rangle\subseteq \langle F_4\rangle$. As a result, the set $\set{\langle F\rangle}{F\in\mathcal{F},F\subseteq U}$ is directed.

 Note that  $\langle F\rangle\subseteq U$ for any $F\subseteq U$, it is clear that $\bigcup\set{\langle F\rangle}{F\in\mathcal{F},F\subseteq U}\subseteq  U$.
  For the reverse inclusion, let $u\in U$. Then there exists some $F_1\in \mathcal{F}$  such that
 $u\in \langle F_1\rangle\subseteq U$. By Definition~\ref{d3.6},
  $u\in \langle F_2\rangle$ and $F_2\subseteq\langle F_1\rangle$ for some $F_2 \in \mathcal{F}$. Thus, $u\in \langle F_2\rangle\subseteq  \bigcup\set{\langle F\rangle}{F\in\mathcal{F},F\sqsubseteq U}$, which means that $U\subseteq \bigcup\set{\langle F\rangle}{F\in\mathcal{F},F\subseteq U}$.

(2) implies (1): %We claim that $U=\langle U \rangle$. In fact, $U=\bigcup\set{\langle F\rangle}{F\in\mathcal{F},F\subseteq U}$ yields that $U \subseteq\langle U \rangle$. Proceeding as in the proof of part~(1), the reverse inclusion also holds.
Assume that $M\sqsubseteq U=\bigcup\set{\langle F\rangle}{F\in\mathcal{F},F\subseteq U}$. Since  $\set{\langle F\rangle}{F\in\mathcal{F},F\subseteq U}$ is directed, we have some $F\in \mathcal{F}$ with $F\sqsubseteq U$ such that $M \sqsubseteq \langle F\rangle$.
\end{proof}

\begin{proposition}\label{p3.13}
If  $(X, \tau\circ\gamma,\mathcal{F})$ is an
 $\mathbb{F}$-augmented generalized closure space, then $\mathcal{R}(X)$ ordered by set inclusion forms a dcpo.
\end{proposition}
\begin{proof}
For any directed family $\set{U_i}{i\in I}$ of $\mathbb{F}$-regular open sets of $(X, \tau\circ\gamma,\mathcal{F})$, let $U=\bigcup{_{i\in I}}U_i$.
Assume that $M\sqsubseteq U$, since $\set{U_i}{i\in I}$ is directed, there exists $j\in I$ such that $M\sqsubseteq U_j$. As $U_j$ is an  $\mathbb{F}$-regular open set of $(X, \tau\circ\gamma,\mathcal{F})$, we get some $F\in \mathcal{F}$ satisfying $M\subseteq \langle F\rangle\subseteq U_j\subseteq U$. Thus $U$ is an  $F$-regular open set of $(X, \tau\circ\gamma,\mathcal{F})$. This implies that the sup $\bigvee_{i\in I}U_i=\bigcup{_{i\in I}}U_i$, and hence $(\mathcal{R}(X),\subseteq)$ is a dcpo.
\end{proof}

The above proposition shows that $(\mathcal{R}(X),\subseteq)$ preserves directed unions. Next, we will characterize the way below relation in term of the composition $\tau\circ\gamma$.
\begin{proposition}\label{p3.14}
Let  $(X, \tau\circ\gamma,\mathcal{F})$ be an
 $\mathbb{F}$-augmented generalized closure space. Then for all $\mathbb{F}$-regular open sets $U_1$  and $U_2$ of $(X, \tau\circ\gamma,\mathcal{F})$,
\begin{equation}\label{e3.7}
U_1\ll U_2\Leftrightarrow (\exists F \in \mathcal{F})(F\subseteq U_2, U_1 \subseteq \langle F\rangle).
\end{equation}
\end{proposition}
\begin{proof}
Assume that $U_1\ll U_2$. Since $U_2$ is an $\mathbb{F}$-regular open set, the set $\set{\langle F\rangle}{F\in \mathcal{F}, F\subseteq U_2}$ is directed and $U_2$ is its union. According to the definition of the way below relation, there exists some $F\in\mathcal{F}$  such that $F\subseteq U_2$ and $U_1 \subseteq \langle F\rangle$.

For the reverse implication, assume that $ U_1 \subseteq \langle F\rangle $ for some $F\in \mathcal{F}$ with $F\subseteq U_2$. Let $\set{V_i}{i\in I}$ be a directed family of $\mathbb{F}$-regular open sets of $(X, \tau\circ\gamma,\mathcal{F})$ such that $U_2\subseteq \bigcup\set{V_i}{i\in I}$. Then $F\sqsubseteq \bigcup\set{V_i}{i\in I}$. So $F\subseteq V_i$ for some $i\in I$. By Definition \ref{d3.10}, we have some $F'\in \mathcal{F}$ such that  $F\subseteq \langle F'\rangle \subseteq  V_i$. This implies that $U_1\subseteq \langle F\rangle\subseteq \langle F'\rangle\subseteq V_i$. As a result, $U_1\ll U_2$.
\end{proof}
\begin{theorem}\label{t3.15}
 Let  $(X, \tau\circ\gamma,\mathcal{F})$ be an
 $\mathbb{F}$-augmented generalized closure space. Then $(\mathcal{R}(X),\subseteq)$ is a continuous domain, which has a basis $\set{\langle F\rangle}{F\in \mathcal{F}}$.

\end{theorem}
\begin{proof}
Proposition \ref{p3.13} has shown that $(\mathcal{R}(X),\subseteq)$ is a dcpo. Then we need only to prove that $(\mathcal{R}(X),\subseteq)$ has a basis.
 Let $U\in \mathcal{R}(X)$. For any $F\in \mathcal{F}$ with $ F\sqsubseteq U$, by Proposition \ref{p3.14}, it is clear that $\langle F\rangle \ll U$. Since $\set{\langle F\rangle}{F\in \mathcal{F}, F\sqsubseteq U}$ is directed and $U=\bigcup \set{\langle F\rangle}{F\in \mathcal{F}, F\sqsubseteq U}$, it follows that
 $\set{\langle F\rangle}{F\in \mathcal{F}}$ is a basis.
\end{proof}

The preceding theorem tells us that each  $\mathbb{F}$-augmented generalized closure space determines a continuous domain.
Conversely, to a continuous domain we may associate an  $\mathbb{F}$-augmented generalized closure space in the way described below.

Given a continuous domain $(D,\leq)$ with a basis $B_D$, for any $A\subseteq B_D$, define
$$\gamma(A)=\da A \cap B_D,\tau(A)=\dda A\cap B_D.$$
 With the interpolation property of the way below, it is trivial to verify that $(B_D,\gamma)$ is a closure space and $(B_D,\tau\circ\gamma)$ is a generalized closure space.
Let $\mathcal{F}_D=\set{F\sqsubseteq B_D}{\bigvee F\in F}$. Then for any $F\in\mathcal{F}_D$, we have $\bigvee F\in D$ and
\begin{equation}
\langle F\rangle=\dda (\bigvee F) \cap B_D.
\end{equation}

 Assume that $M\sqsubseteq \langle F\rangle$. By the interpolation property of the way below relation, there exists some element $d\in B_D$ such that $M\ll d\ll \bigvee F$. Set $\{d\}\cup M=F_1$. Then
$F_1\in \mathcal{F}_D$, $M\subseteq \langle F_1\rangle$ and $M\subseteq F_1\subseteq \langle F\rangle$.

Thus we can summarize what we have discussed as the following proposition.
\begin{proposition}\label{p3.16}
Let $(D,\leq)$ be a continuous domain with a basis $B_D$. Then
$(B_D, \tau\circ\gamma,\mathcal{F}_D)$ is an
 $\mathbb{F}$-augmented generalized closure space.
\end{proposition}
 For any continuous domain~$(D,\leq)$ with a basis~$B_D$, \cite{abramsky94} defined the notion of round ideals. A subset  $U$ of $B_D$ is called a round ideal if it satisfies the following conditions:
\begin{enumerate}[(R1)]
\item $ (\forall u\in B_D )(\forall v\in U)(u\ll v\Rightarrow u\in U)$,
\item $(\forall M\sqsubseteq U)(\exists u\in U)M\ll u$.
\end{enumerate}

 %an $\mathbb{F}$-regular open set of~$(B_D, \tau\circ\gamma,\mathcal{F}_D)$ can be characterized by the way below relation.
\begin{proposition}\label{p3.17}
Let $(D,\leq)$ be a continuous domain with a basis~$B_D$. Then a  subset $U$ of $B_D$ is an $\mathbb{F}$-regular open set of~$(B_D, \tau\circ\gamma,\mathcal{F}_D)$ if and only if it  a round ideal of $B_D$.

\end{proposition}
\begin{proof}

Suppose that $U\in \mathcal{R}(B_D)$. First, let $u\in B_D$ and $v\in U$ with $u\ll v$. By part~(2) of Proposition~\ref{p3.11}, we have $\langle \{v\}\rangle=\dda v\cap B_D\subseteq U$. This implies that $u\in U$. Second, if $M\sqsubseteq U$, then there exists $F\in\mathcal{F}_D$ such that
$M\subseteq\langle F\rangle\subseteq U$. Since $F\in \mathcal{F}_D$, it follows that
$\bigvee F\in F$. Define $u=\bigvee F$, then $M\ll u$.

Conversely, suppose that $U\subseteq B_D$ which satisfies conditions~(R1) and (R2). For any finite subset $M$ of $U$, there exists some $u\in U$ such that $M\ll u$. Set $F=\{u\}$, then $F\in\mathcal{F}_D$ and $M\subseteq \dda u\cap B_D=\langle F\rangle \subseteq U$. Thus $U$ is an $\mathbb{F}$-regular open set of $(B_D, \tau\circ\gamma,\mathcal{F}_D)$.
\end{proof}

 As a directed consequence of Proposition~\ref{p3.17}, each $\mathbb{F}$-regular open set $U$ of~$(B_D, \tau\circ\gamma,\mathcal{F}_D)$ is a directed subset of $D$, and hence, $\bigvee U\in D$. Moreover, it is easy to see that $\dda x\cap B_D$ is a $\mathbb{F}$-regular open set  of~$(B_D, \tau\circ\gamma,\mathcal{F}_D)$ for all $x\in D$. These observations allow us to define functions:
   $$f :D\rightarrow \mathcal{R}(B_D), x \mapsto \dda x\cap B_D,$$
$$g: \mathcal{R}(B_D)\rightarrow D, U \mapsto \bigvee U.$$
 \begin{lemma}The functions $f$ and $g$ defined above are inverse to each other and Scott continuous.
 \end{lemma}

   %With the above preparations, we obtain the main result of this section.
\begin{theorem}[Representation Theorem for Continuous Domains]\label{t3.18}
Let $(D,\leq)$ be a continuous domain with a basis $B_D$. Then $(D,\leq)$  is isomorphic to $(\mathcal{R}(B_D),\subseteq)$.
\end{theorem}
\begin{proof}
Define
$$f :D\rightarrow \mathcal{R}(B_D), x \mapsto \dda x\cap B_D,$$
$$g: \mathcal{R}(B_D)\rightarrow D, U \mapsto \bigvee U.$$
From  Proposition \ref{p3.17}, it follows that $f$ and $g$ are well-defined. And they are obviously  order preserving and mutually inverse. This completes the proof.
\end{proof}

\subsection{The representations of continuous $L$-domains and continuous bounded-complete domains}
In this subsection, we introduce two special classes of $\mathbb{F}$-augmented generalized closure spaces which can be used to present continuous $L$-domains and continuous bounded-complete domains, respectively.
\begin{definition}\label{d3.19}
Given  an
 $\mathbb{F}$-augmented generalized closure space~$(X, \tau\circ\gamma,\mathcal{F})$, let $F\in \mathcal{F}$ and $M\sqsubseteq \langle F\rangle$. An element $G$ of $\mathcal{F}$ is called an $F$-sup of $M$ if it satisfies the following two conditions:
 \begin{enumerate}[({L}1)]
 \item $ \langle M\rangle \subseteq \langle G\rangle$ and $G\subseteq \langle F\rangle$,
     \item for any $G_1\in \mathcal{F}$, $ \langle M\rangle\subseteq\langle G_1\rangle\subseteq\langle F\rangle$ implies $\langle G\rangle\subseteq\langle G_1\rangle$.
 \end{enumerate}
\end{definition}

For any $F\in \mathcal{F}$ with $M\sqsubseteq \langle F\rangle$, the set of all  $F$-sups of $M$ is denoted by $\sum(F,M)$. We first list some basic property of $\sum(F,M)$ in the following.

\begin{proposition}\label{p3.20}
Let $(X, \tau\circ\gamma,\mathcal{F})$ be an $\mathbb{F}$-augmented generalized closure space. Then the following statements hold.
\begin{enumerate}[(1)]
\item If $G_1,G_2\in \sum(F,M)$, then $\langle G_1\rangle=\langle G_2\rangle$.
\item If $M\sqsubseteq \langle F_1\rangle\subseteq \langle F_2\rangle$, then $\sum(F_1,M)\subseteq\sum(F_2,M)$.
\end{enumerate}

Moreover, for any $U\in \mathcal{R}(X)$ and  $F_1,F_2\in \mathcal{F}$ with $F_1,F_2 \subseteq U$.
\begin{enumerate}[(3)]
\item If  $G_1\in\sum(F_1,M)$ and $G_2\in\sum(F_2,M)$, then $\langle G_1\rangle=\langle G_2\rangle$.
    \end{enumerate}
    \begin{enumerate}[(4)]
\item  If $G_1\in\sum(F_1,M_1)$, $G_2\in\sum(F_2,M_2)$ and $ M_1 \subseteq  M_2\subseteq \langle F_1\rangle\cap \langle F_2 \rangle$, then $\langle G_1\rangle\subseteq \langle G_2\rangle$.
\end{enumerate}
\end{proposition}
\begin{proof}
(1) Let $G_1,G_2\in \sum(F,M)$. Then by condition~(L1), it follows that $ \langle M\rangle\subseteq\langle G_1\rangle\subseteq\langle F\rangle$ and $\langle M\rangle\subseteq\langle G_2\rangle\subseteq\langle F\rangle$. According to condition~(L2), we have $\langle G_1\rangle\subseteq\langle G_2\rangle$ and $\langle G_2\rangle\subseteq\langle G_1\rangle$. So that $\langle G_1\rangle=\langle G_2\rangle$.

(2) Let $F_1,F_2\in \mathcal{F}$ with $M\subseteq \langle F_1\rangle$ and $F_1\subseteq \langle F_2\rangle$. Then $\langle F_1\rangle \subseteq \langle F_2\rangle$ and $M\subseteq \langle F_2\rangle$. For any $G_1\in \sum(F_1,M)$ and $G_2\in \sum(F_2,M)$, we claim that $\langle G_1\rangle=\langle G_2\rangle$. Indeed, $G_1\in \sum(F_1,M)$ implies that $ \langle M\rangle \subseteq \langle G_1\rangle \subseteq \langle F_1\rangle \subseteq \langle F_2\rangle$. For $G_2\in \sum (F_2,M)$, using condition (L2), we have $\langle G_2\rangle \subseteq \langle G_1\rangle$. Note that $\langle M\rangle \subseteq \langle G_2\rangle\subseteq\langle G_1\rangle \subseteq \langle F_1\rangle$ and $G_1 \in \sum(F_1,M)$, using condition (L2) again, we have $\langle G_1\rangle \subseteq \langle G_2\rangle$. So that $\langle M\rangle \subseteq \langle G_1\rangle \subseteq\langle G_2\rangle \subseteq \langle F_2\rangle$. Since $G_1\subseteq \langle F_1\rangle\subseteq \langle F_2\rangle$, it follows that $G_1\in \sum(F_2,M) $, and thus $\sum(F_1,M)\subseteq\sum(F_2,M)$.

(3) Since $U$ is an $\mathbb{F}$-regular open set and $F_1\bigcup F_2\sqsubseteq U$, there exists some $F\in\mathcal{F}$ such that $F_1\bigcup F_2 \subseteq \langle F\rangle \subseteq U$. From part (2), we have $\sum(F_1,M)\subseteq \sum(F,M)$ and $\sum(F_2,M)\subseteq \sum(F,M)$. Then  by part (1), $\langle G_1\rangle=\langle G_2\rangle$ for any $G_1\in \sum(F_1,M)$ and $G_2\in \sum(F_1,M)$.

(4) For any $G\in\sum(F_1,M_2)$, we have $\langle M_2\rangle\subseteq\langle G\rangle\subseteq\langle F_1\rangle$. As
$G_2\in\sum(F_2,M_2)$, by part~(3), $\langle G_2\rangle=\langle G\rangle$. Then $\langle M_1\rangle\subseteq \langle G_2\rangle\subseteq\langle F_1\rangle$. Note that $G_1\in \sum(F_1,M_1)$, with condition~(L2), it follows that $\langle G_1\rangle\subseteq\langle G_2\rangle$.
\end{proof}

\begin{definition}\label{d3.21}
 An $\mathbb{F}$-augmented generalized closure space~$(X, \tau\circ\gamma,\mathcal{F})$ is said to be \emph{locally consistent}, if the set $\sum(F,M)$ is nonempty for any $F\in \mathcal{F}$ with $M\sqsubseteq \langle F\rangle$.
\end{definition}

\begin{example}\label{ex3.22}
 Recall  Example \ref{ex3.7}. We claim that $(\mathbb{R},\tau\circ\gamma,\mathcal{F}_{\mathbb{R}})$ is not a locally consistent
  $\mathbb{F}$-augmented generalized closure space.  Indeed, let $F\in \mathcal{F}_{\mathbb{R}}$. Then $\langle F\rangle=(a,+\infty)$, where $a=$inf~$F$. Let $M\sqsubseteq \langle F\rangle$.

 If $ M\neq\emptyset$, then $\sum(F,M)=\set{G\in \mathcal{F}(\mathbb{R})}{\text{inf}~G=m}$, where $m=$inf $M$.

   If $M=\emptyset$, then $\sum(F,\emptyset)=\emptyset$.
\end{example}

\begin{theorem}\label{d3.23}
 Let  $(X, \tau\circ\gamma,\mathcal{F})$ be a locally consistent
 $\mathbb{F}$-augmented generalized closure space. Then $(\mathcal{R}(X),\subseteq)$ is an $L$-domain.
\end{theorem}
\begin{proof}
Theorem~\ref{t3.15} has shown that $(\mathcal{R}(X),\subseteq)$ is a continuous domains. So that it suffices to prove that for any element $U$ of  $\mathcal{R}(X)$, the set $$\da U=\set{V\in \mathcal{R}(X)}{V\subseteq U}$$ is a complete lattice ordered by set inclusion.

We first show that $\da U$ has a least element. As $U$ is an $\mathbb{F}$-regular open set and $\emptyset \sqsubseteq U$, there exists some $F_{\emptyset}\in \mathcal{F}$ such that $\emptyset \subseteq \langle F_{\emptyset}\rangle\subseteq U$. Since $(X, \tau\circ\gamma,\mathcal{F})$ is a locally consistent
 $\mathbb{F}$-augmented generalized closure space, $\sum (F_{\emptyset},\emptyset)\neq \emptyset$. Taking $G_{\emptyset}\in \sum (F_{\emptyset},\emptyset)$, we claim that $\langle G_{\emptyset}\rangle$ is the least element of $\da U$. In fact, suppose that $V$ is an element of $\da U$. Because $\emptyset \sqsubseteq V$, there exists some $F\in \mathcal{F}$ such that $F\sqsubseteq V$ and $\emptyset \subseteq \langle F \rangle$. Since $\emptyset \subseteq \langle F \rangle$ and $F\sqsubseteq U$, by part (4) of proposition~\ref{p3.20}, it follows that $\langle G_{\emptyset}\rangle \subseteq \langle F \rangle$ and hence $\langle G_{\emptyset}\rangle \subseteq V$.

Next,
for any directed subset $\set{V_i}{i\in I}$ of $\da U$, with Proposition~\ref{p3.13}, $\bigcup_{i\in I}V_i\in \mathcal{R}(X)$. Then $\bigcup_{i\in I}V_i\in \da U$, which means that $\da U$ is closed under sups of directed subsets.

 The remainder  is to prove that the least upper bound exists for every pair~$(V_1,V_2)$ of elements of $\da U$.
 Define $$\mathcal{S}=\set{G}{(\exists M\sqsubseteq V_1\cup V_2)(\exists F\sqsubseteq U) G\in \sum(F,M)},$$ and $$V=\bigcup\set{\langle G\rangle}{G\in \mathcal{S}}.$$ We finish the proof by checking that $V$ is the least upper bound of $V_1$ and $V_2$ in $\da U$, which is divided into four steps.

  Step 1, we show $V_1\cup V_2\subseteq V $. Assume that $x\in V_1\cup V_2$. Then $x\in V_1$ or $x\in V_2$. If $x\in V_1$, by Definition~\ref{d3.6}, there exists some $F_{x}\in \mathcal{F}$ such that $\{x\}\subseteq \langle F_{x} \rangle$ and $F_{x} \subseteq V_1 \subseteq U$. As $U$ is an $\mathbb{F}$-regular open set, we get some $G_{x}\in \mathcal{F}$ satisfying $F_{x}\subseteq \langle G_{x} \rangle$ and $G_{x}\subseteq U$. Thus $G\in \mathcal{S}$ for any $G\in \sum(G_x,F_x)$. From $x\in \langle F_{x} \rangle \subseteq \langle G_{x} \rangle$, it follows that $x\in V$. This means that $V_1\sqsubseteq V$. Similarly, $V_2\sqsubseteq V$ and hence $V_1\cup V_2 \sqsubseteq V$.

   Step 2, we show $ V \subseteq U$. For any $G\in \mathcal{S}$, there exist $M_G\in V_1\cup V_2$ and $F_{G}\sqsubseteq U$ such that $G\in \sum(F_G,M_G)$. By condition~(L1), we have  $G\subseteq \langle F_{G}\rangle$. Thus $\langle G \rangle\subseteq \langle F_{G}\rangle \subseteq U$, which implies that $V\subseteq U$.

Step 3, we show that $V\in \mathcal{R}(X)$. As $V=\bigcup\set{\langle G\rangle}{G\in \mathcal{S}}$ and $\langle G\rangle\in\mathcal{R}(X)$ for any $G\in\mathcal{S}$, by Proposition, we need only to verify that $\set{\langle G\rangle}{G\in\mathcal{S}}$ is directed. Suppose that $G_1,G_2\in\mathcal{S}$. Then there exist $M_i\sqsubseteq V_1\cup V_2$ and $F_i\subseteq U$ such that $G_i\in \sum(F_i,M_i), i=1,2$. Since $M_1\cup M_2 \sqsubseteq V_1\cup V_2\subseteq U$, we get some $F_3\in \mathcal{F}$ with $F_3\subseteq U$ and $M_1\cup M_2 \sqsubseteq \langle F_3\rangle$. Because $(B, \tau\circ\gamma,\mathcal{F})$ is a
locally consistent $F$-augmented generalized closure space, $\sum(F_3,M_1\cup M_2)\neq \emptyset$. Taking $G_3\in\sum(F_3,M_1\cup M_2)$,
  it is clear that $G_3\in\mathcal{\mathcal{S}}$. Since $ M_1, M_2 \subseteq  M_1\cup M_2 \subseteq \langle F_3 \rangle$,  as the proof of part~(4) of Proposition~\ref{p3.20}, we have that $\langle G_1 \rangle \subseteq \langle G_3 \rangle$ and $\langle G_2 \rangle \subseteq \langle G_3 \rangle$.

Step 4, we have to show that for any upper bound $V_3$ of $V_1$ and $V_2$ in $\da U$, the inclusion $V\subseteq V_3$ holds.  For this, let $G\in \mathcal{S}$. Then there exist $M_G\sqsubseteq V_1\cup V_2$ and $F_G\subseteq U$ such that
$G\in\sum(F_G,M_G)$. This implies that $M_G\subseteq V_3$.  By part~(2) of Proposition~\ref{p3.11}, we get some $F\in \mathcal{F}$
 such that $F\subseteq V_3$ and $M_G\subseteq \langle F\rangle$. So that $\langle M_G\rangle \subseteq \langle F\rangle$. Since $F, F_G\subseteq U$, with part~(3) of Proposition~\ref{p3.20},
$\langle G\rangle =\langle G_1\rangle$ for any $G_1\in \sum(F,M_G)$. As a result, $\langle G\rangle \subseteq\langle F\rangle\subseteq V_3$, and hence $V\subseteq V_3$.
\end{proof}
\begin{theorem} [Representation Theorem for Continuous $L$-domains]\label{t3.24}
Let $(D,\leq)$ be an  $L$-domain with a basis $B_D$. Then
$(B_D, \tau\circ\gamma,\mathcal{F}_D)$ is a
locally consistent $\mathbb{F}$-augmented generalized closure space, and  $(D,\leq)$  is order isomorphic to $(\mathcal{R}(B_D),\subseteq)$.
\end{theorem}
\begin{proof}
By Proposition~\ref{p3.16}, $(B_D, \tau\circ\gamma,\mathcal{F}_D)$ is an
 $\mathbb{F}$-augmented generalized closure space. And by Theorem~\ref{t3.18}, $(\mathcal{R}(B_D),\subseteq)$  is isomorphic to $(D,\leq)$. Then it suffices to prove that $(B_D, \tau\circ\gamma,\mathcal{F}_D)$ is locally consistent, that is $\sum(F,M)\neq\emptyset$ for any $F\in\mathcal{F}_D$ and $M\sqsubseteq \langle F\rangle$.
 From the definition of $\mathcal{F}_D$, we know that $\bigvee F\in F\subseteq B_D$. Then $\langle F\rangle=\dda F\cap B_D=\dda \bigvee F\cap B_D$. Since $(D,\leq)$ is a continuous  $L$-domain, $\da \bigvee F$ has a least element $a_F\in B_D$
 and $a_F\ll \bigvee F$.

 If $M=\emptyset$, let $G=\{a_F\}$. Then
$$ \langle M\rangle=\emptyset\subseteq \dda a_F \cap B_D=\langle G\rangle\text{~and~} G=\{a_F\} \subseteq\dda \bigvee F \cap B_D=\langle F\rangle.$$
For any $G_1\in \mathcal{F}$ with $ M\subseteq \langle G_1\rangle\subseteq\langle F\rangle$, we have $a_F\in\dda \bigvee G_1\cap B_D\subseteq\dda \bigvee F\cap B_D$. So that $\langle G\rangle=\dda a_F \cap B_D\subseteq\dda \bigvee G_1\cap B_D=\langle G_1\rangle$.

 On the other hand, suppose that $M\neq\emptyset$, let $G=\{m_F\}$, where $m_F$ is the sup of $M$ in $\da \bigvee F$. Since  $M\subseteq\dda \bigvee F \cap B_D$, we have $G\subseteq\dda \bigvee F \cap B_D$. Then
 $$\langle M\rangle =\dda  M\cap B_D \subseteq \dda m_F\cap B_D =\langle G\rangle\text{~and~} G=\{m_F\} \subseteq\dda \bigvee F \cap B_D=\langle F\rangle.$$
And suppose that $G_1\in \mathcal{F}$ with $\langle M\rangle \subseteq \langle G_1\rangle\subseteq\langle F\rangle$. Then $m_F\leq \bigvee G_1\leq \bigvee F$. This implies that
 $\dda m_F \cap B_D\subseteq \dda \bigvee G_1 \cap B_D$, and hence $\langle G\rangle\subseteq\langle G_1\rangle$.
\end{proof}

In the rest of this section, we give a representation of continuous bounded-complete domains.
\begin{definition}\label{d3.25}
 An $\mathbb{F}$-augmented generalized closure space~$(X, \tau\circ\gamma,\mathcal{F})$ is said to be \emph{consistent}, if it satisfies the following condition,
 \begin{enumerate}
 %\item [(BC)]  $ (\forall F\in\mathcal{F},\forall M\sqsubseteq X) M\sqsubseteq \langle F\rangle \Rightarrow (\exists x_M\in \langle F\rangle) M\cup \{x_M\}\in \mathcal{F}$.
 \item [(BC)]  for any $F\in\mathcal{F}$ and $M\sqsubseteq \langle F\rangle$, there is a unique  $x_M\in \langle F\rangle$ such that $ M\cup \{x_M\}\in \mathcal{F}$.
 \end{enumerate}
\end{definition}
\begin{proposition}
Each consistent  $\mathbb{F}$-augmented generalized closure space is locally consistent.
\end{proposition}
\begin{proof}
Suppose that~$(X, \tau\circ\gamma,\mathcal{F})$ is a  consistent  $\mathbb{F}$-augmented generalized closure space. If $F\in \mathcal{F}$ and $M\sqsubseteq \langle F\rangle$, then
 there is a unique  $x_M\in \langle F\rangle$ such that $ M\cup \{x_M\}\in \mathcal{F}$. Let $G=M\cup \{x_M\}$.
It is clear that $G\in \sum(F,M)$. Thus,  $(X, \tau\circ\gamma,\mathcal{F})$ is a locally  consistent  $\mathbb{F}$-augmented generalized closure space.
\end{proof}
\begin{theorem}  [Representation Theorem for Continuous bounded-complete Domains]\label{t3.27}

 Suppose that $(X, \tau\circ\gamma,\mathcal{F})$ is a consistent $\mathbb{F}$-augmented generalized closure space, then $(\mathcal{R}(X),\subseteq)$ is a continuous bounded-complete domain.

Conversely, let $(D,\leq)$ be a continuous bounded-complete domain with a basis $B_D$. Then
$(B_D, \tau\circ\gamma,\mathcal{F}_D)$ is a consistent $\mathbb{F}$-augmented generalized closure space and $(\mathcal{R}(X),\subseteq)$ is isomorphic to $(D,\leq)$.
\end{theorem}
\begin{proof}
  We first prove $(\mathcal{R}(X),\subseteq)$ is a continuous bounded-complete domain. By Theorem~\ref{t3.15}, it suffices to verify that any two
$\mathbb{F}$-regular open sets which are bounded above have a sup.
 Let $U_1,U_2$ and $U$ be $\mathbb{F}$-regular open sets with $U_1,U_2\subseteq U$. Set
 \begin{equation} \label{equa3.3}
 V=\bigcup\set{\langle F\rangle}{(\exists M\sqsubseteq U_1\cup U_2,\exists x_M\in U)(M\cup \{x_M\}\in \mathcal{F})}.
 \end{equation}
 We now show that $V$ is also an $\mathbb{F}$-regular open set and that it is the sup of $U_1$ and $U_2$.

 Suppose $M\sqsubseteq V$. For any $m\in M$, by equation (\ref{equa3.3}), there exists some $F_m\in \mathcal{F}$ such that   $m\in \langle F_{m}\rangle$ and $F_m\subseteq U_1\cup U_2$. Then $M\subseteq \langle\bigcup_{m\in M} F_{m}\rangle$ and $F_m\sqsubseteq U$. Because $U$ is an $\mathbb{F}$-regular open set and $\bigcup_{m\in M}F_{m}\sqsubseteq U$, we have some $F\in \mathcal{F}$ with $\bigcup_{m\in M}F_{m}\sqsubseteq  \langle F\rangle\subseteq U$. Set $N=\bigcup_{m\in M}F_{m}$. Since $(X, \tau\circ\gamma,\mathcal{F})$ is consistent, there is some $x_N\in \langle F\rangle$ such that $N\cup \{x_N\}\in \mathcal{F}$. Note that $$M\subseteq \langle N\cup \{x_N\}\rangle\sqsubseteq \langle F\rangle\subseteq U_1\cup U_2\subseteq V,$$ it follow that $V$ is an $\mathbb{F}$-regular open set. Moreover, it is clear that $V_1,V_2\subseteq V$ by part~(2) of Proposition~\ref{p3.11}. Let $U_3$ be any other $\mathbb{F}$-regular open set with $U_1,U_2\subseteq U_3$. Then $U_1\cup U_2\subseteq U_3$, and thus $V\subseteq U_3$ by the definition of $V$.

 For the reverse direction, by Theorem~\ref{t3.18}, we need only to show that $(B_D, \tau\circ\gamma,\mathcal{F}_D)$ satisfies condition~(BC). In fact, let $F\in \mathcal{F}_D$ and $M\sqsubseteq \langle F\rangle$. Then $\bigvee F \in F$ and $M\cup F\subseteq \da F$. Since $(D,\leq)$ is  bounded complete, we have $\bigvee (M\cup F)\in M\cup F$. This implies that $M\cup F\in \mathcal{F}_D$.
\end{proof}

\section{Category equivalence}

 In the previous section, we have investigated the representation of continuous domains by $\mathbb{F}$-augmented generalized closure spaces. From the categorical viewpoint, we have only provided the object part of a functor.
In this section, we aim to extend this relation to a categorical equivalence. On the side of continuous domains, one typically uses Scott-continuous functions as morphisms to build a category~{\bf CD}. So  we have to introduce an appropriate notion of morphisms for $\mathbb{F}$-augmented generalized closure spaces which can be used to represent Scott-continuous functions between continuous domains.

\begin{definition}\label{d4.1}
Let  $(X, \tau\circ\gamma,\mathcal{F})$ and $(X', \tau'\circ\gamma',\mathcal{F'})$  be two
 $\mathbb{F}$-augmented generalized closure spaces. A  relation~$\Theta\subseteq\mathcal{F}\times X'$ is called an  \emph{approximable mapping} from $(X, \tau\circ\gamma,\mathcal{F})$  to $(X', \tau'\circ\gamma',\mathcal{F'})$ if the following conditions hold:
 \begin{enumerate}[({AM}1)]
\item $F\Theta F'\Rightarrow F\Theta \langle F'\rangle$,
 \item $F\sqsubseteq \langle F_1\rangle, F\Theta M'\Rightarrow F_1\Theta M'$,
 \item $F\Theta M'\Rightarrow (\exists G\in \mathcal{F},G'\in \mathcal{F'})(G\subseteq\langle F\rangle , M'\subseteq\langle G'\rangle, G\Theta G')$,
 \end{enumerate}
  for any $F,F_1\in \mathcal{F}$, $F'\in\mathcal{F}'$ and $M'\sqsubseteq X'$,
where $F\Theta M'$ means that $F\Theta x'$ for any $x'\in M'$. This situation is denoted by writing  $\Theta:X\rightarrow X'$.
\end{definition}

We now present some basic properties of  approximable mappings that will be used in the sequel.
\begin{proposition}\label{p4.2}
Let $\Theta$ be an approximable mapping  from $(X, \tau\circ\gamma,\mathcal{F})$  to $(X', \tau'\circ\gamma',\mathcal{F'})$. For any $F,F_1\in \mathcal{F}$ and $M'\sqsubseteq X'$, the following statements hold.
\begin{enumerate}[(1)]
\item  $F\Theta M'$ if and only if there exists some $ G\subseteq\langle F\rangle$ such that $G\Theta M'$.
\item If $F\Theta M'$, then there exists some $G'\in \mathcal{F'}$ such that $M'\subseteq\langle G'\rangle$ and $ F\Theta G'$.
\item If $F_1\sqsubseteq F$ and $F_1\Theta M'$, then $F\Theta M'$.

\end{enumerate}
\end{proposition}
\begin{proof}

(1) Let $F\Theta M'$. Then by condition~(AM3), there exist $G\in\mathcal{F} $ and $G'\in \mathcal{F'}$ such that $G \subseteq\langle F\rangle$, $G\Theta G'$ and $M'\subseteq \langle G'\rangle$. From $G\Theta G'$, using condition~(AM1), we have  $G\Theta \langle G'\rangle$. Since $M'\subseteq \langle G'\rangle$, it follows that $G\Theta M'$. The reverse implication is clear by condition~(AM2).

(2) Assume that $F\Theta M'$, then by condition~(AM3), there exists $G\in\mathcal{F} $ and $G'\in \mathcal{F'}$ such that $G \subseteq\langle F\rangle$, $G\Theta G'$ and $M'\subseteq \langle G'\rangle$. From $G \subseteq\langle F\rangle$ and $G\Theta G'$, by condition~(AM2), we have $F\Theta G'$.

 (3) Let $F_1\sqsubseteq F$ and $F_1\Theta M'$. By part (1), there exists $G\sqsubseteq\langle F_1\rangle$ such that $G\Theta M'$. From $F_1 \sqsubseteq F$, it follows that $G\sqsubseteq\langle F_1\rangle\subseteq\langle F\rangle$.
  Thus by condition~(AM2), we have that $F\Theta M'$.
\end{proof}

Given two $\mathbb{F}$-augmented generalized closure spaces $(X, \tau\circ\gamma,\mathcal{F})$ and $(X', \tau'\circ\gamma',\mathcal{F'})$, let $\Theta$ be
an approximable mapping   from $(X, \tau\circ\gamma,\mathcal{F})$  to $(X', \tau'\circ\gamma',\mathcal{F'})$.
 For any $F\in \mathcal{F}$, we write
\begin{equation}\label{e4.1}
\Theta(F)=\set{x\in X'}{F\Theta x'}.
\end{equation}
And for any $\mathbb{F}$-regular open set~$U$ of  $(X,\tau\circ\gamma,\mathcal{F})$, we write
\begin{equation}\label{e4.2}
\Theta(U)=\set{x'\in X'}{(\exists F\in \mathcal{F})(F\sqsubseteq U,F\Theta x')}.
\end{equation}
\begin{proposition}\label{p4.3}
Let $\Theta$ be an approximable mapping  from $(X, \tau\circ\gamma,\mathcal{F})$  to $(X', \tau'\circ\gamma',\mathcal{F'})$.

\begin{enumerate}[(1)]
\item
For any $F\in \mathcal{F}$, $\Theta(F)$ is an $\mathbb{F}$-regular open set of  $(X', \tau'\circ\gamma',\mathcal{F'})$.
\item For any $\mathbb{F}$-regular open set~$U$ of  $(X, \tau\circ\gamma,\mathcal{F})$, $\Theta(U)$ is an $\mathbb{F}$-regular open set of  $(X', \tau'\circ\gamma',\mathcal{F'})$.

\end{enumerate}
\end{proposition}
\begin{proof}
(1) For any $F'\in \mathcal{F}'$, by part~(1) of Proposition~\ref{p3.11}, $\langle F'\rangle$ is an $\mathbb{F}$-regular open set of  $(X', \tau'\circ\gamma',\mathcal{F'})$. Then by Proposition~\ref{p3.13}, it suffices to show that the set $\set{\langle F'\rangle}{F'\in \mathcal{F'},F\Theta F'}$ is directed and $\Theta(F)$ is its union.

 We first prove $$\Theta(F)=\bigcup \set{\langle F'\rangle}{F'\in \mathcal{F'},F\Theta F'}.$$ For any $x'\in \Theta(F)$, by part (2) of Proposition~\ref{p4.2}, there exists $F'\in \mathcal{F}'$ such that $\{x'\}\subseteq \langle F'\rangle$ and $F\Theta F'$. Using condition~(AM1), it follows that $F\Theta \langle F'\rangle$. Thus
  $x'\in \bigcup \set{\langle F'\rangle}{F'\in \mathcal{F'},F\Theta F'}$.
   This implies that
  $ \Theta(F)\subseteq\bigcup \set{\langle F'\rangle}{F'\in \mathcal{F'},F\Theta F'}$. Conversely, let $x'\in \bigcup\set{\langle F'\rangle}{F'\in \mathcal{F'},F\Theta F'}$. That is, $x'\in \langle F'\rangle $ for some $ F'\in \mathcal{F'}$ with $F\Theta F'$. From condition (AM1), it follows that $F\Theta \langle F'\rangle $. Then $x'\in \Theta(F)$. So that $\bigcup \set{\langle F'\rangle}{F'\in \mathcal{F'},F\Theta F'}\subseteq \Theta(F)$.

We now claim that $\set{\langle F'\rangle}{F'\in \mathcal{F'},F\Theta F'}$ is directed. In fact,
let $F'_1$ and $F'_2$ be two elements of $\mathcal{F'}$ such that  $F\Theta F_1'$ and $F\Theta F_2'$. Then $F\Theta( F_1'\cup F_2')$. By part (2) of Proposition~\ref{p4.2}, there exists $F_3'\in\mathcal{F'}$ such that $F_1'\cup F_2' \subseteq \langle F_3'\rangle$  and $F\Theta F_3'$. Since $\langle F_3'\rangle$ is an $\mathbb{F}$-regular open set of $(X', \tau'\circ\gamma',\mathcal{F'})$, with part (2) of Proposition~\ref{p3.11}, there exists $F_4'\in \mathcal{F'}$ such that $F_1'\cup F_2' \subseteq \langle F_4'\rangle$ and $F_4'\subseteq \langle F_3'\rangle$. From $F_1'\cup F_2' \subseteq \langle F_4'\rangle$, we have $\langle F_1'\rangle \subseteq \langle F_4'\rangle$ and $\langle F_2'\rangle \subseteq \langle F_4'\rangle$. From $F_4'\subseteq \langle F_3'\rangle$ and $F\Theta F_3'$, using condition~(AM1), we have $F\Theta F_4$.  As a result, $\set{\langle F'\rangle}{F'\in \mathcal{F'},F\Theta F'}$ is directed.

 (2)  Assume that $M'\sqsubseteq \Theta(U)$. Then for any $m'\in M'$, there exists some $F_{m'}\in\mathcal{F}$ such that $F_{m'}\subseteq U$ and $F_{m'}\Theta m'$. Since $\bigcup_{m'\in M'}F_{m'}\sqsubseteq U$, with part (2) of Proposition~\ref{p3.11}, we get some $F\in\mathcal{F}$ such that $F\subseteq U$ and $\bigcup_{m'\in M'}F_{m'}\subseteq \langle F\rangle$.  Thus $F\Theta M'$ using condition~(AM2).
 By part~(2) of Proposition~\ref{p4.2}, there exists some $G'\in \mathcal{F'}$ such that $M'\subseteq \langle G'\rangle$ and $F\Theta G'$. So that $M'\subseteq \langle G'\rangle\subseteq \Theta(U)$, which completes the proof.
\end{proof}
The above proposition guarantees that the approximable mapping~$\Theta$ assigns to every $F\in \mathcal{F}$   an $\mathbb{F}$-regular open set of $(X', \tau'\circ\gamma',\mathcal{F'})$, and also shows that $\Theta$ can provide a passage from
$\mathbb{F}$-regular open sets of $(X, \tau\circ\gamma,\mathcal{F})$ to those of $(X', \tau'\circ\gamma',\mathcal{F'})$.

We now turn to investigate how Scott-continuous functions between continuous domains can be represented by the notion of approximable mapping.
\begin{theorem}\label{t4.4}
  Let $(X,\tau\circ\gamma,\mathcal{F})$ and $(X',\tau'\circ\gamma',\mathcal{F'})$ be two $\mathbb{F}$-augmented generalized closure spaces. Given an approximable mapping $\Theta$ from $(X,\tau\circ\gamma,\mathcal{F})$ to $(X',\tau'\circ\gamma',\mathcal{F'})$, define an assignment $\phi_{\Theta}: \mathcal{R}(X)\rightarrow \mathcal{R}(X')$ by
   \begin{equation}\label{e4.3}
 \phi_{\Theta}(U) = \set{x'\in X' }{ (\exists F\in\mathcal{F})(F\subseteq U, F\Theta x')}.
 \end{equation}
  Then  $\phi_{\Theta}$
is a Scott-continuous function from $(\mathcal{R}(X),\subseteq)$ to  $(\mathcal{R}(X'),\subseteq)$.

Conversely, if $\phi$ is a Scott-continuous function from $(\mathcal{R}(X),\subseteq)$ to  $(\mathcal{R}(X'),\subseteq)$, define a relation  $\Theta_{\phi}\subseteq\mathcal{F}\times X'$ by
 \begin{equation}\label{e4.4}
 F \Theta_{\phi} x'\Leftrightarrow x'\in \phi(\langle F\rangle).
 \end{equation}
Then $\Theta_{\phi}$
is an approximable mapping from $(X,\tau\circ\gamma,\mathcal{F})$ to  $(X',\tau'\circ\gamma',\mathcal{F'})$.

Moreover, $\phi_{{\Theta}_{\phi}}=\phi$ and $\Theta_{\phi_{\Theta}}=\Theta$.
\end{theorem}
\begin{proof}
 Let $\Theta$ be an approximable mapping  from $(X,\tau\circ\gamma,\mathcal{F})$ to  $(X,\tau'\circ\gamma',\mathcal{F'})$.
With part~(2) of Proposition~\ref{p4.3}, it follows that the function $\phi_{\Theta}$ is well-defined. For any directed subset $\set{U_i}{i\in I}$ of $\mathcal{R}(X)$, since $\phi_{\Theta}$ is clearly order-preserving, it follows that $\set{\phi_{\Theta}(U_i)}{i\in I}$ is a directed subset  of $\mathcal{R}(X')$. From Proposition \ref{p3.13}, we know that $\bigvee_{i\in I}U_i=\bigcup_{i\in I}U_i$ and $\bigvee_{i\in I}\phi_{\Theta}(U_i)=\bigcup_{i\in I}\phi_{\Theta}(U_i)$.
 So that to prove  $\phi_{\Theta}$ is Scott-continuous, since $\bigcup_{i\in I}\phi_{\Theta}(U_i)\subseteq \phi_{\Theta}(\bigcup_{i\in I}U_i)$ is obvious, it suffices to
show that the reverse inclusion holds. Assume that $x'\in \phi_{\Theta}(\bigcup_{i\in I}U_i)$, then there exists some $F\in \mathcal{F}$ such that $F\sqsubseteq \bigcup_{i\in I}U_i$ and $F \Theta x'$, which implies that $F\sqsubseteq U_j$ for some $j\in I$. Thus $x'\in  \phi_{\Theta}(U_j)\subseteq\bigcup_{i\in I}\phi_{\Theta}(U_i)$. As a result, $\phi_{\Theta}(\bigcup_{i\in I}U_i)\subseteq\bigcup_{i\in I}\phi_{\Theta}(U_i)$.

Let $\phi$ be a Scott-continuous function from $(\mathcal{R}(X),\subseteq)$ to  $(\mathcal{R}(X'),\subseteq)$. We now prove that $\Theta_{\phi}$ is an approximable mapping by checking the three conditions in Definition \ref{d4.1}. Let  $F,F_1\in \mathcal{F}$, $F'\in\mathcal{F}'$ and $M'\sqsubseteq X'$.
 Assume that $F\Theta_{\phi} F'$, that is $F'\subseteq \phi(\langle F\rangle)$. Since $\langle F\rangle\in \mathcal{R}(X)$, it follows that $\phi(\langle F\rangle)\in \mathcal{R}(X')$. This implies that $\langle F'\rangle \subseteq \phi(\langle F\rangle)$. By equation (\ref{e4.4}), $F\Theta_{\phi} \langle F'\rangle $. Condition~(AM1) follows.
For condition~(AM2), assume that $F\sqsubseteq \langle F_1\rangle$ and $F \Theta_{\phi}M'$. Then $\langle F\rangle\subseteq\langle F_1\rangle$ and hence $M'\subseteq \phi(\langle F\rangle)\subseteq \phi (\langle F_1\rangle)$. Therefore, $F_1\Theta_{\phi}M'$.
To prove condition~(AM3), assume that $F \Theta_{\phi}M'$. Then $M'\sqsubseteq\phi(\langle F\rangle)$. Since $\langle F\rangle$ is the directed union of the set $\set{\langle G\rangle}{G\in\mathcal{F},G\sqsubseteq \langle F\rangle}$ and $\phi$ is Scott-continuous, it follows that $$\phi(\langle F\rangle)=\phi(\bigcup \set{\langle G\rangle}{G\in\mathcal{F},G\sqsubseteq \langle F\rangle})=\bigcup \set{\phi(\langle G\rangle)}{G\in\mathcal{F},G\sqsubseteq \langle F\rangle}.$$ So that $M'\sqsubseteq \phi(\langle G\rangle)$ for some $G\in\mathcal{F}$ with $G\sqsubseteq \langle F\rangle$. For $M'\sqsubseteq \phi(\langle G\rangle)$, using part~(2) of Proposition \ref{p3.11}, we have some $G'\in \mathcal{F'}$  satisfying $G'\subseteq \phi(\langle G\rangle)$ and $M'\sqsubseteq \langle G'\rangle$. To sum up, there exist some $G\in\mathcal{F}$
and $G'\in \mathcal{F'}$ such that $G\sqsubseteq \langle F\rangle$,  $M'\sqsubseteq\langle G'\rangle$ and $G\Theta_{\phi}G'$.

For any $U \in \mathcal{R}(X)$, we have
\begin{align*}
\phi_{\Theta_{\phi}}(U)&=\set{x'\in X'}{(\exists F\in\mathcal{F})(F\sqsubseteq U,F\Theta_{\phi}x')}\\
&=\set{x'\in X'}{(\exists F\in\mathcal{F})(F\sqsubseteq U,x\in\phi(\langle F\rangle))}\\
&=\bigcup \set{\phi(\langle F\rangle)}{\exists F\in\mathcal{F},F\sqsubseteq U}\\
&=\phi(\bigcup\set{\langle F\rangle}{\exists F\in\mathcal{F},F\sqsubseteq U})\\
&=\phi(U).
\end{align*}
This proves that $\phi_{\Theta_{\phi}}=\phi$.

And for any $F\sqsubseteq X$ and $x'\in X'$, we have
\begin{align*}
F\Theta_{\phi_{\Theta}}x'&\Leftrightarrow x'\in \phi_{\Theta}(\langle F\rangle)\\
&\Leftrightarrow (\exists G\in \mathcal{F})(G\subseteq \langle F\rangle, G\Theta x')\\
&\Leftrightarrow F\Theta x'.
\end{align*}
This proves that $\Theta_{\phi_{\Theta}}=\Theta$.
\end{proof}

The above theorem shows that there is a one-to-one correspondence between approximable mappings from $(X,\tau\circ\gamma,\mathcal{F})$ to $(X',\tau'\circ\gamma',\mathcal{F'})$ and Scott-continuous functions from $\mathcal{R}(X)$ to $\mathcal{R}(X')$.
The other way around, suppose that $(D,\leq)$ and $(D',\leq')$ are continuous domains. Proposition~\ref{p3.16} has shown that  $(B_D,\tau\circ\gamma,\mathcal{F}_D)$ and $(B_{D'},\tau'\circ\gamma',\mathcal{F'}_{D'})$ are $\mathbb{F}$-augmented generalized closure spaces. We next discuss the relationship between Scott continuous functions from $(D,\leq)$ to $(D',\leq')$ and approximable mappings from $(B_D,\tau\circ\gamma,\mathcal{F}_D)$ to  $(B_{D'},\tau'\circ\gamma',\mathcal{F'}_{D'})$. To this end, we need the following lemma.
\begin{lemma}\label{l4.5}
 Let $\Phi$ be an approximable mapping from $(B_D,\tau\circ\gamma,\mathcal{F}_D)$ to $(B_{D'},\tau'\circ\gamma',\mathcal{F'}_{D'})$. For any $x\in D$, define
 \begin{equation}\label{e4.5}
  I_x= \set{x'\in D' }{ (\exists F\in \mathcal{F}_D)(F\subseteq \dda x\cap B_D, F\Phi x')}.
 \end{equation}
 Then $I_x $ has a {\em{sup}} in $D'$.
\end{lemma}
\begin{proof}
For any $x\in D$, by Proposition~\ref{p3.17}, it is easy to prove that $\dda x\cap B_D$ is an $\mathbb{F}$-regular open set of $(B_D,\tau\circ\gamma,\mathcal{F}_D)$. With part~(2) of Proposition~\ref{p4.3}, we see that $I_x=\Phi(\dda x\cap B_D)\in\mathcal{R}(X')$. This implies that $I_x$ is a directed subset of $D'$, and hence $\bigvee I_x \in D'$.
\end{proof}

\begin{theorem}\label{t4.6}
  Let $(D,\leq)$  be a continuous domain with a basis~$B_D$ and $(D',\leq')$   a continuous domain with a basis~$B_{D'}$.  For any Scott continuous function $f : D \rightarrow D'$, define a relation $\Phi_f\subseteq\mathcal{F}_D\times D'$ by
  \begin{equation}\label{e4.6}
   F\Phi_f x'\Leftrightarrow  x'\ll' f(\bigvee F).
   \end{equation}
  Then $\Phi_f$
is an approximable mapping from $(B_D,\tau\circ\gamma,\mathcal{F}_D)$ to  $(B_{D'},\tau'\circ\gamma',\mathcal{F'}_{D'})$.

Conversely, if $\Phi$ is an approximable mapping from $(B_D,\tau\circ\gamma,\mathcal{F}_D)$ to $(B_{D'},\tau'\circ\gamma',\mathcal{F'}_{D'})$, then by setting
 \begin{equation}\label{e4.7}
 f_{\Phi}(x) = \bigvee\set{x'\in D' }{ (\exists F\in \mathcal{F}_D)(F\subseteq \dda x\cap B_D, F\Phi x')}.
 \end{equation}
One obtains a Scott continuous function $f_{\Phi} : D \rightarrow D'$.

Moreover, $f=f_{{\Phi}_f}$ and $\Phi=\Phi_{f_{\Phi}}$.
\end{theorem}
\begin{proof} For any Scott-continuous function $f : D \rightarrow D'$, we check that the relation defined by equation~(\ref{e4.6}) is an approximable mapping from $(B_D,\tau\circ\gamma,\mathcal{F}_D)$ to $(B_{D'},\tau'\circ\gamma',\mathcal{F'}_{D'})$ in the following:

For (AM1), suppose that $F'\in \mathcal{F}_{D'}'$ and $F\Phi_f F'$. Then $F'\ll' f(\bigvee F)$. Since $F'$ is finite and $\bigvee F'\in D'$, it follows that $\bigvee F'\ll'f(\bigvee F)$. That is $\langle F'\rangle\ll'f(\bigvee F)$. This implies that $F\Phi_f \langle F'\rangle$.

For (AM2), let $F\sqsubseteq \langle F_1\rangle$ and $F\Phi_f M'$, where $M'$ is a finite subset of $D'$. Then $F\ll \bigvee F_1$ and $M'\ll'f(\bigvee F)$. As $f$ is order-preserving, $M'\ll'f(\bigvee F)$. This means that  $F_1\Phi_f M'$.

 For (AM3), suppose that $F\Phi_f M'$ with $M\sqsubseteq D'$. Then $M'\ll' f(\bigvee F)$. By the interpolation property of $\ll'$, there exists some $d'\in B_{D'}$ such that $M'\ll' d'\ll' f(\bigvee F)$.  Note that $f(\bigvee F)=f(\bigvee (\dda (\bigvee F)))=\bigvee (\dda f(\bigvee F))$, we have some $d\in \dda (\bigvee F)\cap B_D$ with $d'\ll' f(d)$. Set $G=\{d\}$ and $G'=\{d'\}\cup M'$. Thus $G\in \mathcal{F}$ and $G'\in \mathcal{F'}$ such that $G\sqsubseteq \langle F\rangle$, $M'\subseteq \langle G'\rangle$ with $G\Phi_f G'$.

 Given an approximable mapping  $\Phi$ from $(B_D,\tau\circ\gamma,\mathcal{F}_D)$ to  $(B_{D'},\tau'\circ\gamma',\mathcal{F'}_{D'})$. By Lemma~\ref{l4.5}, we see that the function~$f_{\Phi}$ defined by equation~(\ref{e4.7}) is well-defined and  $f_{\Phi}(x)=I_x$, for any $x\in D$. We now prove that $f_{\Phi}$ is Scott-continuous by checking that $f_{\Phi}(\bigvee S)=\bigvee f_{\Phi} (S)$ for any directed subset $S$ of $D$.

   In fact it is clear that $I_x\subseteq I_y$ for any $x,y \in D$ with $x\leq y$. So that $f_{\Phi}$ is order-preserving, and then $\bigvee f_{\Phi} (S)\leq f_{\Phi}(\bigvee S)$.
   Conversely, since $\bigvee f_{\Phi} (S)= \bigvee \set{\bigvee I_d}{d\in S}=\bigvee (\bigcup{_{d\in S}}I_d)$ and $f_{\Phi}(\bigvee S)=\bigvee I_{\bigvee S}$, to complete the proof, it suffices to show that $I_{\bigvee S}\subseteq\bigcup{_{d\in S}}I_d$. Let $x'\in  I_{\bigvee S}$. Then $F \Phi x'$ for some $F\in \mathcal{F}$ with $F\sqsubseteq \dda \bigvee S\cap B_D$. Because $D'$ is a continuous domain, $\bigvee F\ll' \bigvee S$. As $S$ is directed, there exists some $d\in S$ with $\bigvee F\ll' d$, which implies that $F\sqsubseteq \dda d \cap B_D$. Thus $x'\in I_d$, and then $I_{\bigvee S}\subseteq\bigcup{_{d\in S}}I_d$.

 For any $x\in D$, we have
 \begin{align*}
f_{\phi_{f}}(x)&=\bigvee\set{x'\in B_{D'}}{(\exists F\in \mathcal{F})(F\sqsubseteq \dda x\cap B_D, F\phi_{f}x'})\\
&=\bigvee\set{x'\in B_{D'}}{(\exists F\in \mathcal{F})(F\sqsubseteq \dda x\cap B_D,  x'\ll'f(\bigvee F)}\\
&=\bigvee \set{x'\in B_{D'}}{(\exists y\in B_D)( y\ll x, x'\ll'f(y))}\\
&=\bigvee (\dda f(x)\cap B_{D'})\\
&=f(x)
\end{align*}
 This implies that $f=f_{\phi_{f}}$.

 For any $F\sqsubseteq \mathcal{F}_D$ and $x'\in D'$, we have
\begin{align*}
F\Phi_{f_{\Phi}}x'&\Leftrightarrow x'\ll'f_{\Phi}(\bigvee F)\\
&\Leftrightarrow x'\ll'\bigvee\set{y'\in L'}{(\exists F_1\in\mathcal{F}_D)(F_1\subseteq \dda \bigvee F\cap B_D, F\Phi y'})\\
&\Leftrightarrow (\exists y'\in D', \exists F_1\in\mathcal{F}_D)(F_1\sqsubseteq \dda \bigvee F\cap B_D,F\Phi y', x'\ll y')\\
%&\Leftrightarrow (\exists \{y'\}\sqsubseteq D', F_1\in\mathcal{F}_D)(F_1\sqsubseteq \dda \bigvee F\cap B_D,F\Phi\{y'\}, \{x'\}\sqsubseteq \langle y'\rangle)\\
&\Leftrightarrow ( \exists F_1\in\mathcal{F}_D)(F_1\sqsubseteq\langle F\rangle,F_1\Phi x')\\
&\Leftrightarrow F\Phi x'.
\end{align*}
 This implies that $\Phi=\Phi_{f_{\Phi}}$.
\end{proof}
\begin{proposition}\label{p4.7}
$\mathbb{F}$-augmented generalized closure spaces with approximable mappings form a category~${\bf FGC}$.
\end{proposition}
\begin{proof}
Let $\Theta$ be an
approximable mapping from $(X,\tau\circ\gamma,\mathcal{F})$ to $(X',\tau'\circ\gamma',\mathcal{F'})$ and $\Theta'$ be an
approximable mapping from $(X',\tau'\circ\gamma',\mathcal{F'})$ to $(X'',\tau''\circ\gamma'',\mathcal{F''})$. Define $\Theta\circ \Theta'\subseteq \mathcal{F}\times X''$ by
\begin{equation}\label{e4.8}
F(\Theta\circ \Theta')x''\Leftrightarrow (\exists G\in\mathcal{F'})(F\Theta G,G\Theta'x''),
\end{equation}
and id$_X\subseteq \mathcal{F}\times X$ by
\begin{equation}\label{e4.9}
F\text{id}_X x\Leftrightarrow x\in \langle F\rangle.
\end{equation}
Routine checks verify that $\Theta\circ \Theta'$ is an approximable mapping from $(X,\tau\circ\gamma,\mathcal{F})$ to $(X'',\tau''\circ\gamma'',\mathcal{F''})$ and
 $\text{id}_X$ is an approximable mapping from $(X,\tau\circ\gamma,\mathcal{F})$ to itself.

 Conditions (1) and (2) of Definition~\ref{d4.1} yield that $\text{id}_X$ is the identity morphism of $(X,\tau\circ\gamma,\mathcal{F})$. Using the same argument as checking the associative law of a traditional relation composition, we can easy carry out $\circ$ defined by equation~(\ref{e4.8}) is also associative.
\end{proof}
 Analogously, we can establish a category~{\bf LFGC} of locally consistent $\mathbb{F}$-augmented generalized  closure spaces and a category~{\bf CFGC} of  consistent $\mathbb{F}$-augmented generalized  closure spaces. They are full subcategories of ${\bf FGC}$.

\begin{lemma}\label{l4.8}
$\mathcal{G}: {\bf FGC}\rightarrow {\bf CD}$ is a functor which maps every $\mathbb{F}$-augmented generalized  closure space $(X,\tau\circ\gamma,\mathcal{F})$ to $(\mathcal{R}(X),\subseteq)$ and approximable mapping $\Theta:X\rightarrow X'$ to $\phi_{\Theta}: \mathcal{R}(X)\rightarrow \mathcal{R}(X')$, where $\phi_{\Theta}$ is defined by equation~(\ref{e4.3}).
\end{lemma}
\begin{proof}
By Theorems~\ref{t3.15} and~\ref{t4.4}, $\mathcal{G}$ is well-defined.
For any $U\in\mathcal{R}(X)$, we have
\begin{align*}
\mathcal{G}(\text{id}_{X})(U)&=\phi_{\text{id}_X}(U) \\
&=\set{x\in X}{(\exists F\in \mathcal{F})(F\sqsubseteq U, x\in \langle F\rangle})\\
&=\bigcup\set{\langle F\rangle}{(\exists F\in \mathcal{F})(F\sqsubseteq U)}\\
&=U.
\end{align*}
This implies that $\mathcal{G}$ preserves the identity morphism.

Let $\Theta: X\rightarrow X', \Theta': X'\rightarrow X''$ be two approximable mappings.
For any $U\in\mathcal{R}(X)$ and $x''\in X''$, we have
\begin{align*}
x''\in \mathcal{G}(\Theta'\circ\Theta)(U)&\Leftrightarrow x''\in f_{\Theta'\circ\Theta}(U) \\
&\Leftrightarrow (\exists F\in \mathcal{F})(F\sqsubseteq U,F(\Theta'\circ\Theta)x'' )\\
&\Leftrightarrow ( F\in \mathcal{F},\exists G\in \mathcal{F'})(F\sqsubseteq U,F\Theta G,G\Theta'x'')\\
&\Leftrightarrow (\exists G\in \mathcal{F'})( G\sqsubseteq \phi_{\Theta}(U),G\Theta'x'')\\
&\Leftrightarrow (\exists G\in \mathcal{F'})( G\sqsubseteq  \phi{\Theta}(U),G\Theta' x'')\\
&\Leftrightarrow x''\in \phi_{\Theta'}(\mathcal{G}(\Theta)(U))\\
&\Leftrightarrow x''\in\mathcal{G}(\Theta')(\mathcal{G}(\Theta)(U)).
\end{align*}
This implies that $\mathcal{G}(\Theta'\circ\Theta)=\mathcal{G}(\Theta')\circ\mathcal{G}(\Theta)$, and then $\mathcal{G}$ preserves the composition.
\end{proof}
\begin{remark}\label{4.9}
Corresponding to the functor $\mathcal{G}: {\bf FGC}\rightarrow {\bf CD}$, we can also define a functor~$\mathcal{H}$ from ${\bf CD}$ to ${\bf FGC}$ as follows: for any continuous domain~$(D,\leq)$ with a basis~$B_D$, $$\mathcal{H}(D)=(B_D,\tau\circ\gamma,\mathcal{F}_D),$$
and for any Scott-continuous functions~
$f : D \rightarrow D'$, $$\mathcal{H}(f)=\Phi_f,$$
where $\Phi_f$ is defined  by equation~(\ref{e4.6}). With Proposition~\ref{p3.16}
and Theorem~\ref{t4.6},  the above two functions are well-defined. Using a similar way as in the proof of Lemma~\ref{l4.8}, one can show that $\mathcal{H}$ is a functor.
\end{remark}
Instead of directly proving there exists a pair of natural isomorphisms we use an alternative method to show the equivalence of categories ${\bf FGC}$ and ${\bf CD}$.
\begin{lemma}\label{l4.10}\cite{maclane71}
The following conditions on the categories ${\bf C}$ and ${\bf D}$ are equivalent:
\begin{enumerate}[(1)]
\item  ${\bf C}$ and ${\bf D}$ are categorically equivalent.
 \item There exists a functor~$\mathcal{G}:{\bf C}\rightarrow{\bf D}$ such that $\mathcal{G}$ is full, faithful and essentially surjective on objects, that is for every object~$D$ of ${\bf D}$, there exists some object $C$ of ${\bf C}$ such that $\mathcal{G}(C)\cong D$.
 \end{enumerate}
\end{lemma}
\begin{theorem}\label{t4.11}
${\bf FGC}$ and ${\bf CD}$ are categorically equivalent.
\end{theorem}
\begin{proof}
According to Theorem~\ref{t3.18}, it suffices to show that the functor~$\mathcal{G}$ defined in Lemma~\ref{l4.8} is full and faithful.

Let $\phi:\mathcal{R}(X)\rightarrow \mathcal{R}(X')$ be a Scott-continuous function. Define a relation  $\Theta_{\phi}\subseteq\mathcal{F}(X)\times X'$ by
 \begin{equation}\label{e4.10}
 F \Theta_{\phi} x'\Leftrightarrow x'\in \phi(\langle F\rangle).
 \end{equation}
Then  by Theorem~\ref{t4.4} , $\Theta_{\phi}$
is an approximable mapping from $(X,\tau\circ\gamma,\mathcal{F})$ to  $(X',\tau'\circ\gamma',\mathcal{F'})$ and $\mathcal{G}(\Theta_{\phi})=\phi_{\Theta_{\phi}}=\phi$.
This implies that $\mathcal{G}$ is full.

 Let $\Theta_1,\Theta_2: X \rightarrow X'$ be two approximable mappings with $\phi_{\Theta_1}=\phi_{\Theta_2}$. For any $F\in \mathcal{F}$, since
\begin{align*}
\set{x'\in X'}{F\Theta_1 x'}&=\set{x'\in X'}{(\exists G\in\mathcal{F})(G\sqsubseteq F ,G\Theta_1 x')}\\
%&=\bigcup\set{\langle G\rangle}{(\exists G\in\mathcal{F})(G\sqsubseteq F )}\\
&=\phi_{\Theta_1}(\langle F\rangle)\\
&=\phi_{\Theta_2}(\langle F\rangle)\\
%&=\bigcup\set{\langle G\rangle}{(\exists G\in\mathcal{F})(G\sqsubseteq F )}\\
&=\set{x'\in X'}{(\exists G\in\mathcal{F})(G\sqsubseteq F ,G\Theta_1 x')}\\
&=\set{x'\in X'}{F\Theta_2 x'},
\end{align*}
%$$F\Theta_1x''\Leftrightarrow x''\in f_{\Theta_1}(\langle F\rangle)\Leftrightarrow x''\in f_{\Theta_2}(\langle F\rangle)\Leftrightarrow F\Theta_2x''$$
it follows that $\Theta_1=\Theta_2$, and hence $\mathcal{G}$ is faithful.
\end{proof}
So far, we have established the equivalence between the category of $\mathbb{F}$-augmented generalized closure spaces and that of continuous domains. This result suggests a novel approach to representing continuous domains by means of closure spaces.

It is worth noting that the category~${\bf LD}$ of continuous $L$-domains and the category~${\bf BCD}$ of  continuous bounded-complete domains are two full subcategories of ${\bf FGC}$. Then based on Theorems~\ref{t3.24} and~\ref{t3.27}, we  have
(1) ${\bf LFGC}$ and ${\bf LD}$ are categorically equivalent, and
(2)  ${\bf CFGC}$ and ${\bf BCD}$ are categorically equivalent.
Consequently, both ${\bf LFGC}$ and ${\bf CFGC}$ are Cartesian closed categories.

\bibliographystyle{plain}

\end{document}